\documentclass[a4paper,12pt]{article}
\usepackage[math]{kurier}
\usepackage[T1]{fontenc}
\usepackage{graphicx}
\usepackage{enumitem}
\usepackage[
  top=1.5cm,
  left=1.5cm,
  right=1.5cm,
  bottom=2cm
]{geometry}

\usepackage{amsmath}
\usepackage{amsthm}
\usepackage{mathtools}
\usepackage{tikz}
\usepackage{amsfonts}
\usepackage[hidelinks]{hyperref}

\usepackage{ulem}
\usepackage{contour}
\contourlength{0.8pt}

\newcommand{\myuline}[1]{
  \uline{\phantom{#1}}%
  \llap{\contour{white}{#1}}%
}
\renewcommand{\emph}[1]{\myuline{#1}}

\newtheorem{lemma}{Lemma}
\newtheorem{proposition}{Proposition}
\newtheorem{definition}{Definition}
\newtheorem{remark}{Remark}

\usepackage{tikz}
\usetikzlibrary{arrows,fit,calc}
\usepackage{xspace}
\usepackage{amsmath}
\usepackage{amssymb}
\usepackage{mathpartir}
\usepackage{mathtools}
\usepackage{bbm}

\DeclareSymbolFont{stmry}{U}{stmry}{m}{n}
\DeclareMathDelimiter\llbracket{\mathopen}{stmry}{"4A}{stmry}{"71}
\DeclareMathDelimiter\rrbracket{\mathclose}{stmry}{"4B}{stmry}{"79}
\DeclareMathSymbol\leftrightarrowtriangle\mathbin{stmry}{"5D}
\DeclareMathSymbol\leftarrowtriangle\mathrel{stmry}{"5E}
\DeclareMathSymbol\rightarrowtriangle\mathrel{stmry}{"5F}
\DeclareMathSymbol\inplus\mathrel{stmry}{"41}
\DeclareMathSymbol\niplus\mathrel{stmry}{"42}

\NewDocumentCommand\paren{r()}{\left(#1\right)}
\NewDocumentCommand\Paren{r()}{\left(\hspace{-2pt}\middle( #1\middle)\hspace{-2pt}\right)}
\RenewDocumentCommand\brack{r[]}{\left[#1\right]}
\newcommand\set[1]{\left\{#1\right\}}
\newcommand\setcompr[2]{\left\{#1~\middle|~#2\right\}}
\NewDocumentCommand\tuple{r<>}{\left\langle#1\right\rangle}
\NewDocumentCommand\Tuple{r<>}{\left\langle\hspace{-2pt}\middle\langle #1\middle\rangle\hspace{-2pt}\right\rangle}
\NewDocumentCommand\sem{r[]}{\left\llbracket {#1}\right\rrbracket}
\NewDocumentCommand\variables{r()}{\left\lfloor #1 \right\rfloor}
\NewDocumentCommand\size{r||}{\left|#1\right|}

\NewDocumentCommand\Functor{ m s d()}{%
    \ensuremath{\IfValueTF{#3}{\IfBooleanTF{#2}{{#1}{{#3}}}{{#1}\paren(#3)}}{#1}}\xspace%
}

\newcommand\f[1]{\Functor{\mathcal{#1}}}
\NewDocumentCommand\relarrow{ O{.08} r() O{midway, above} o r() }{
  \IfNoValueTF{#4}{
    \draw[->](#2) to (#5);
  }{
    \draw[->](#2) to node[#3] {#4} (#5);
  }
  \draw ($(#2)!.5!(#5)$) circle (#1)
}
\newcommand\relation[1][.5]{
  \mathrel{
    \tikz{
      \coordinate(s) at (0,0);
      \coordinate(t) at (#1,0);
      \relarrow(s)(t);
    }
  }
}

\newcommand\natrel[1][.5]{
  \mathrel{
    \tikz{
      \coordinate(s) at (0,0);
      \coordinate(t) at (#1,0);
      \draw[]($(s.east) + (0,.03)$) to ($(t.west) - (.03,-.03)$);
      \draw[]($(s.east) - (0,.03)$) to ($(t.west) - (.03,.03)$);
      \draw[]($(t.west) + (-.06,.06)$) -- (t.west) -- ($(t.west) - (.06,.06)$);

      \draw ($(s)!.5!(t)$) circle (.08);
    }
  }
}

\newcommand\converse{^{\circ}}
\newcommand\under{\mathbin{\setminus}}
\renewcommand\over{\mathbin{/}}
\newcommand\then{\mathbin{;}}

\newcommand\eqrel{~\leftrightarrowtriangle~}
\newcommand\leqrel{~\rightarrowtriangle~}

\newcommand\neqrel{~\mathclap{~\;\,\slash}{\leftrightarrowtriangle}~}
\newcommand\nleqrel{~\not\rightarrowtriangle~}

\newcommand\mmodels{\mathrel{|\hspace{-.1em}{\models}}}

\newcommand\trleq{\mathrel{\preccurlyeq}}

\newcommand\Reg{\ensuremath{\mathrm{Reg}}\xspace}
\newcommand\KA{\ensuremath{\mathrm{KA}}\xspace}
\newcommand\Set{\ensuremath{\mathbf{Set}}\xspace}
\newcommand\Rel{\ensuremath{\mathbf{Rel}}\xspace}
\newcommand\Repr{\ensuremath{\mathbf{Repr}}\xspace}
\newcommand\xRepr{\ensuremath{\mathbf{xRepr}}\xspace}
\newcommand\PreO{\ensuremath{\mathbf{PreOrd}}\xspace}
\newcommand\pset{\f{P}}

\newcommand{\Mon}{\mathrm{Mon}}
\newcommand{\List}{\mathrm{List}}
\newcommand{\Pom}{\mathrm{Pom}}

\newcommand\eqdef{\mathrel{{\mathop:}=}}

\newcommand\id{\mathop{\mathnormal{id}}}
\newcommand\unit{\mathbbm{1}}

\NewDocumentCommand\univ{ s r() }{%
\IfBooleanTF{#1}{{#2}^{\forall}}{{\paren(#2)}^{\forall}}%
}

\renewcommand{\epsilon}{\varepsilon}
\renewcommand{\phi}{\varphi}

\newcommand\R{\f R}
\newcommand\E{\f E}
\newcommand\T{\f T}

\NewDocumentCommand\var{r()}{{\mathrm{var}}\paren(#1)}

\newcommand{\M}{\mathcal M}
\newcommand{\V}{\mathcal V}
\newcommand{\Alg}{\mathfrak{A}}

\title{
{\Huge\scshape Representations}\\
{\large\itshape A meta-model for system  analysis}\\
\vspace{1cm}
{\normalsize EPISEN \& LACL, Université Paris-Est Créteil Val de Marne\\
\href{mailto:paul@brunet-zamansky.fr}{\tt paul@brunet-zamansky.fr}\\
\url{https://paul.brunet-zamansky.fr}
}
}
\author{Paul Brunet}

\begin{document}
\maketitle

\begin{abstract}
The formal analysis of automated systems is an important and growing industry.
This activity routinely requires new verification frameworks to be developed to tackle new programming features, or new considerations (bugs of interest).
Often, one particular property can prove frustrating to establish: completeness of the logic with respect to the semantics.
In this paper, we try and make such developments easier, with a particular attention on completeness.
Towards that aim, we propose a formal (meta-)model of software analysis systems (SAS), the eponymous \emph{Representations}.
This model requires few assumptions on the SAS being modeled, and as such is able to capture a large class of such systems.
We then show how our approach can be fruitful, both to understand how existing completeness proofs can be structured, and to leverage this structure to build new systems and prove their completeness.
\end{abstract}

\section{Introduction}
\label{sec:intro}
As automated systems are given more numerous and varied tasks, so do formal analysis questions grow in importance and diversity.
When in still recent times the field was dominated by the study of critical systems and their potential failures, we see nowadays questions related to e.g. quality of service being tackled by formal analysis.
Here both theory and practice agree: new questions require new systems of analysis.
Indeed, theorems like~\cite{rice_classes_1953,godel_uber_1931} can be interpreted as suggesting that any formal system has either limited expressivity or intractable algorithmic properties.
And so, new SAS (software analysis systems) are being routinely defined and studied.

We consider that the essence of such a system is given by a semantic relation, together with a logic of some sort to reason on programs.
The correctness of such a system hinges on two properties: soundness and completeness.
Soundness states that any valid judgment in the logic is satisfied by the semantics.
Completeness on the other hand requires that any semantically sound judgment can be established by the logic.
It is no surprise then that completeness is usually the difficult bit in correctness proofs: while soundness can be established by checking the soundness of individual rules of the logic, completeness asks the prover to produce a derivation for every true semantic fact, for which no general method seems applicable.

Our approach builds on category theory and relation algebra to provide a point-free theory of representations.
We recall relevant elements of relation algebra in Section~\ref{sec:prelim}.
This allows for proofs that use few assumptions, and constructions that proceed in a modular fashion.
Concretely, we define in Section~\ref{sec:repr} a \emph{representation} to be a structure
\[R=\tuple<T,E,{\models},{\leq}>\]
where $T$ and $E$ are abstract sets, modeling respectively the semantic and syntactic domain of a SAS.
${\models}$ is a binary relation, linking semantic elements to syntactic elements, i.e. a subset of $T\times E$.
The final ingredient is a preorder ${\leq}$ over $E$, modeling the valid derivations in the logic of the SAS under consideration.
To formulate soundness and completeness of a representation structure, we turn to relation algebra:
\begin{align}
 {\models}\then{\leq}&\leqrel{\models}\label{intro:ax:sound}\\
 {\models}\under{\models}&\leqrel{\leq}\label{intro:ax:compl}
\end{align}
Axiom \eqref{intro:ax:sound} models soundness: translated as a first-order (F.O.) sentence, it states that whenever we have $t\models e$ and $e\leq f$, we also have $t\models f$.
This means that whenever the logic derives $e\leq f$, every element in the semantic interpretation of $e$ is also in the semantic interpretation of $f$.
Conversely, axiom \eqref{intro:ax:compl} states completeness: two elements $e,f\in E$ are related by the left-hand side of the entailment iff for every $t\in T$, $t\models e$ implies $t\models f$.
The right-hand side of the axiom then deduces from that fact that $e\leq f$. This accurately models the requirement that any true semantic inclusion is derivable in the logic.

Having set up this model, we may then try and study the theory of representations to derive results about SAS.
The contributions of this paper can be described in this light as follows:
\begin{itemize}
 \item
 In the remainder of Section~\ref{sec:repr}, we study the transformations between representations.
 We define a category of representations, with a notion of homomorphism between representations.
 We show this category to be Cartesian and consider the properties of the product.
  We then look at proof techniques to establish completeness.
  We propose a deductively complete template for such proofs, and show how several proofs from the literature can be cast as instances of our template.
 \item
 In Section~\ref{sec:hor}, we lift our attention to a higher perspective, by considering parametric representations.
 These are modeled by functors from the category of sets and functions to that of representations.
 We characterize these objects, and show how they can be used to define multi-tiered representations.
\end{itemize}
We believe this will allow future completeness proofs to be better structured, easier to follow, and more succinct, as some proof obligations may be discharged thanks to general results.
It is our hope that this work provides the basis for a modular meta-system, able to generate sound and complete SAS in a transparent fashion.
Such a tool would allow formal analysis techniques to be applied to a broader class of systems and class of questions about them.

Because of the author's personal experience, examples are mostly chosen from the field of extensions of Kleene algebra.
Nevertheless, we hope we provide enough intuitions to convince the reader that our approach is not restricted to the study of these formalisms, but can instead capture a large class of formalisms.

\newpage
\section{Preliminaries}
\label{sec:prelim}
We assume the reader is familiar with basic concepts from category theory and relation algebra. We recall below some of the latter, mainly to fix notations.

\paragraph{Relations} A binary relation $x:A\relation B$ is a subset of $A\times B$. To avoid ambiguity, (containment being a relation of interest), we shall denote the set-theoretic inclusion of relations by $\leqrel$, writing $x\leqrel y$ instead of $x\subseteq y$. Similarly, equality of relations is denoted by $\eqrel$. The union and intersection of relations will be written using the standard $\cup$ and $\cap$ symbols.
Given sets $A,B,C$, and relations $x:A\relation B$ and $y:B\relation C$, we define the following constants and operations: $1_A:A\relation A$ is the identity (diagonal) relation over $A$, $x\converse:B\relation A$ is the converse of $x$, and $x\then y:A\relation C$ is the sequential composition of $x$ and $y$, containing all pairs $(a,c)$ such that there is some $b\in B$ such that $(a,b)\in x$ and $(b,c)\in y$. If $s:A\relation A$ is a relation, its reflexive-transitive closure is written $s^\star$. Given $z:A\to C$, the (left-)residual of $z$ by $x$, written $x\under z:B\relation C$, is the set of pairs $(b,c)$ such that for every $a\in A$, $(a,b)\in x$ entails $(a,c)\in z$. This operator can also be defined by the following universal equivalence:
\begin{equation}
y\leqrel x\under z~\Leftrightarrow~ x\then y\leqrel z.\label{eq:galois}                                                                 \end{equation}
A dual operator $z\over y$, satisfying $x\leqrel z\over y$ iff $x\then y\leqrel z$ can be defined by $z\over y\eqdef\paren((y\converse)\under(z\converse))\converse$.

\paragraph{Functions} Given functions $f:A\to B$ and $g:B\to C$, the composition $g\circ f:A\to C$ is the function mapping an $a\in A$ to $g(f(a))$. The identity function over $A$ is written $\id_A$. Functions may be injectively mapped to relations by considering the graph of the function. We write $f_\ast:A\relation B$ for the relation $\setcompr{(a,f(a))\in A\times B}{a\in A}$. Its converse is written $f^\ast:B\relation A$. Functions can be identified among relations by two properties: univalence ($f^\ast\then f_\ast\leqrel 1_B$) and totality ($1_A\leqrel f_\ast\then f^\ast$).
Indeed, any function satisfies both, and any univalent and total relation is the graph of some function.

We will make use of the following property of functions and residuals~\cite[Exercise 4.35]{bird_algebra_1997}:
\begin{equation}
f_\ast\then(x\under y)\then g^\ast\eqrel(x\then f^\ast)\under (y\then g^\ast).\label{eq:fun-res}
\end{equation}


\paragraph{Products and coproducts}
The Cartesian product of two sets $A_1$ and $A_2$ is written $A_1\times A_2$, with the two projections written as $\pi_i:A_1\times A_2\to A_i$. The coproduct of $A_1$ and $A_2$ is written $A_1+A_2$, with the two injections $\iota_i:A_i\to A_1+A_2$.
The coproduct can be fully axiomatized in relation algebra by the following properties:
\begin{itemize}
 \item the $\iota_i$s are injective: $1_{A_i}\eqrel \iota_{i\ast}\then \iota_i^\ast$;
 \item their codomains are disjoint: $\iota_{1\ast}\then \iota_2^\ast\eqrel 0$ (or equivalently $\iota_{2\ast}\then \iota_1^\ast\eqrel 0$);
 \item their codomains cover $A_1+A_2$: $1_{A_1+A_2}\leqrel \iota_1^\ast\then \iota_{1\ast}\cup \iota_2^\ast\then \iota_{2\ast}$.
\end{itemize}

\paragraph{Preorders}
A preorder $x:A\relation A$ are usually defined as reflexive and transitive relations.
This can be written as $1_A\leqrel x$ and $x\then x\leqrel x$.
Equivalently, a preorder is a fixpoint of the reflexive-transitive closure: $x\eqrel x^\star$, or a fixpoint of the self-residual: $x\eqrel x\under x$.
As an example, notice that containment between sets is the self-residual of the membership relation : $\in\under\in\eqrel\subseteq$.

\newpage
\section{The category of representations}
\label{sec:repr}

\subsection{The model \& some of its instances}
The protagonist in this paper is the notion of \emph{representation}.
It is a crude but rather general formalism to describe analysis frameworks.
More specifically it aims at capturing systems where syntactic specifications and combinatorial behaviours are related by a satisfaction relation.
Another relation, over specifications, is there to reason about semantic containment.
Formally, we set:
\begin{definition}[representation]
  A \emph{representation} is a tuple $R=\tuple<T,E,\models,{\leq}>$ where $T$ is a set of \emph{traces}, $E$ a set of \emph{expressions}, $\models:T\relation E$ is a \emph{satisfaction} relation, and ${\leq}$ is a preorder on $E$ such that ${\models}\then{\leq}\leqrel\models$. It is called \emph{exact} when $\paren({\models}\under{\models})\leqrel{\leq}$.
\end{definition}

For any such structure, we define an interpretation function $I:E\to \pset T$ (where $\pset$ is the powerset functor) by setting $I(e)\eqdef\setcompr{t\in T}{t\models e}$. This entails the relational identity: ${\in}\then I^*\eqrel {\models}$. With this notation, the requirements for being a representation and for being exact may look more familiar:
\begin{align*}
  \paren({\models}\then{\leq}\leqrel\models)
  &\Leftrightarrow \forall e,f,\; e\leq f\text{ entails }I(e)\subseteq I(f)\\
  \paren({\models}\under{\models}\leqrel{\leq})
  &\Leftrightarrow \forall e,f,\; I(e)\subseteq I(f)\text{ entails }e\leq f.
\end{align*}
Indeed, the first axiom is soundness of ${\leq}$ with respect to $I$, and the second is completeness.

\begin{remark}
In particular, we have that $(e,f)\in({\models}\under{\models})$ if and only if $I(e)\subseteq I(f)$.
Thus, the relation ${\models}\under{\models}$ is semantic containment. A natural question is then what is being represented by the relation ${\models}\over{\models}$ ?
By definition we have:
\begin{align*}
 (s,t)\in({\models}\over{\models})
 &\Leftrightarrow \forall e,\;t\models e\Rightarrow s\models e.
\end{align*}
In other words, every expression satisfied by $t$ is also satisfied by $s$.
This is a preorder on traces, and often plays a role in the semantic domain, although not in the present paper.  
\end{remark}
\begin{remark}
  Equivalently, we could have started with an interpretation $I: E\to \pset T$, and defined ${\models}$ as ${\in}\then I^\ast$.
  Soundness of ${\leq}$ with respect to $I$ would entail the axiom ${\models}\then{\leq}\leqrel{\models}$, and its completeness would mean the representation is exact.
\end{remark}

Examples of representations abound.
We now briefly examine how concepts from logical semantics and universal algebra may be reformulated in our setting.

\paragraph{Logics}
Consider a logic $L$, given as a set of formulas $\Phi_L$ and an entailment relation $\vdash_L$.
Let $\M=\tuple<\size|\M|,{\models}_\M>$ be a structure such that ${\models}_\M:\size|\M|\relation\Phi_L$ is a satisfaction relation linking semantic elements to individual formulas.
Such a structure is called a \emph{model} of the logic $L$ iff whenever $\phi\vdash_L \psi$ and $m\models_\M \phi$, we also have $m\models_\M\psi$.
Therefore, for every model of the logic, the structure $\tuple<\size|M|,\Phi_L,{\models}_\M,\vdash_L>$ is a representation.
It is exact when the model is fully-abstract, that is when from the sentence ``for all $m\in \size|\M|$, $m\models_\M \phi$ only if $m\models_\M \psi$'' we can deduce the existence of a valid entailment $\phi\vdash_L\psi$.

\paragraph{Varieties}
Consider a functional signature $\Sigma$, a set of variables $\V$ and a (finite) set of axioms $\E:T_\Sigma(\V)\relation T_\Sigma(\V)$ (where $T_\Sigma(\V)$ is the set of terms built from variables in $\V$ and operators in $\Sigma$).
This generates a class of $\Sigma$-algebras known as a \emph{variety}.
Consider such an algebra $\Alg$, and an interpretation map $I:\V\to\size|\Alg|$.
We may build a representation in this case as $\tuple<\size|\Alg|,T_\Sigma(\V),I^{\Alg\ast},\equiv_{\E}>$, where $I^{\Alg}:T_\Sigma(\V)\to\size|\Alg|$ is the homomorphic extension of $I$, and $\equiv_{\E}$ is the smallest congruence containing $\E$.
Note that in this case the preorder is an equivalence relation.
The above representation is exact exactly in situations where the algebra $\Alg$ is free in the variety generated by $\Sigma$ and $\E$.
This point of view may be refined in the case where $\size|\Alg|$ has the structure of a set of sets, i.e. there is a domain (set) $D$ such that $\size|\Alg|\subseteq \pset D$.
In this case, we prefer the following representation as a model:
\[\tuple<D,T_\Sigma(\V),{\in}\then I^{\Alg\ast},\equiv_{\E}>.\]
This works nicely e.g. when $\Alg$ is the complex algebra (powerset algebra) associated with another $\Sigma$-algebra.
\paragraph{Application to Kleene algebras}
While not a finitely axiomatisable variety, Kleene algebra may be understood as a variety generated by infinitely many identities~\cite{krob_complete_1990,salomaa_two_1966}, or a quasi-variety generated by a finite set of Horn sentences~\cite{kozen_completeness_1994}.
Either way, we get (exact) representations, as the details of the axiomatization being used is of little consequence to our model.
Let $\Reg(A)$ be the set of regular expressions (i.e. using concatenation, union ($+$), singleton empty word, empty language ($0$), and Kleene star) over some finite alphabet $A$.
We assume some axiomatization \KA, and consider the relation $e\leq_\KA f$ on $\Reg(A)$ defined as $\KA\vdash e+f=f$, where the later judgment denotes the equational theory defined by \KA (either style of axiomatization can be used here).
As we sketched in the previous paragraph, any model of \KA, i.e. a structure $\Alg$ that can interpret each operator of \Reg and such that any valid judgment $\KA\vdash e=f$ entails an identity $I^\Alg(e)=I^\Alg(f)$ for any $I:A\to\size|\Alg|$, can be cast as a representation.
Furthermore, the free Kleene algebra yields an exact representation $\KA(A)\eqdef\tuple<A^\star,\Reg(A),{\models}_\KA,{\leq}_\KA>$ when $w\models_\KA e$ is defined as ``$w$ belongs to the rational language associated with $e$''.
The same construction yields exact representations when applied to the free models of commutative Kleene algebra~\cite{pilling_algebra_1970}, concurrent Kleene algebra~\cite{kappe_concurrent_2018}, synchronous Kleene algebra~\cite{wagemaker_completeness_2019}, Kleene algebra with tests~\cite{kozen_kleene_1997} and many other extensions of Kleene algebra. 
As our development makes few assumptions on the details of the underlying set theory, it is our belief that nominal extensions of Kleene algebra~\cite{brunet_formal_2016,kozen_completeness_2015} would also be amenable to such treatment.

\paragraph{Trivial representations}
It will prove useful to notice that any relation $x:A\relation B$ generates an exact representation $R^x\eqdef\tuple<A,B,x,x\under x>$.
We dub this the \emph{trivial representation} associated with $x$.
An interesting instance of this construction arises from the membership relation. Let $A$ be a set, the representation
\[R^\in\eqdef\tuple<A,\pset A,\in,\subseteq>.\]
is always exact.
Furthermore, since for every exact representation it holds that ${\leq}\eqrel{\models}\under{\models}$, an exact representation is always isomorphic to the trivial representation associated with its satisfaction relation.

\subsection{Morphisms and products}
In this section we define the category \Repr and establish it as Cartesian.
\begin{definition}[morphism]
 Given two representations $R_1$ and $R_2$, a morphism from the former to the latter is a pair $\tuple<\phi,\psi >:R_1\to R_2$ such that:
 \begin{itemize}
  \item $\phi:E_1\to E_2$ is a function and $\psi :T_2\relation T_1$ is a relation;
  \item $\phi$ is order preserving: $\phi^\ast\then{\leq}_1\leqrel{\leq}_2\then\phi^\ast$;
  \item the interpretation of $\phi(e)$ can be computed from that of $e\in E_1$ using $\psi $:
  \[{\models}_2\then\phi^\ast\eqrel\psi\then{\models}_1.\]
 \end{itemize}
\end{definition}
The last condition is expressed in F.O. style as
\[\forall s\in T_2,\,e\in E_1,\;s\models_2\phi(e)\Leftrightarrow \exists t\in T_1:s\mathrel{\psi}t\wedge t\models_1 e.\]
Or, using set comprehension:
\[I_2(\phi(e))=\setcompr{s}{\exists t\in T_1:s\mathrel{\psi}t\wedge t\in I_1(e)}.\]

It is easy to check that for every representation $\tuple<\id_E,1_T>$ is a morphism, and that given $\tuple<\phi,\psi >:R_1\to R_2$ and  $\tuple<\phi',\psi '>:R_2\to R_3$, $\tuple<\phi'\circ \phi,\psi'\then\psi >$ is a morphism from $R_1$ to $R_3$. We may thus define \Repr as the category whose objects are representations and arrows are morphisms.
We also define $\xRepr$, the full-sub category of \Repr consisting of exact representations and morphisms between them.

\Repr is a Cartesian category, with the product of two representations $R_1$ and $R_2$ defined as follows:
\[R_1\times R_2\eqdef\tuple<T_1+T_2,E_1\times E_2,{\mmodels},{\leqq}>\]
where ${\mmodels}\eqrel \iota_1^\ast\then{\models}_1\then\pi_1^\ast\cup \iota_2^\ast\then{\models}_2\then \pi_2^\ast$, and ${\leqq}\eqrel\pi_{1\ast}\then{\leq}_1\then\pi_1^\ast\cap\pi_{2\ast}\then{\leq}_2\then\pi_2^\ast$.
In other words, we have $I(e_1,e_2)\simeq I_1(e_1)\uplus I_2(e_2)$, and $(e_1,e_2)\leqq(f_1,f_2)$ iff $e_i\leq_if_i$ for both $i\in\set{1,2}$. As stated in the following proposition, this construction preserves exactness, i.e. the category $\xRepr$ is also Cartesian.

\begin{proposition}\label{lem:exact-product}
   $R_1\times R_2$ is the Cartesian product of $R_1$ and $R_2$ in $\Repr$.
  Furthermore, if both $R_1$ and $R_2$ are exact, so is their product.
\end{proposition}
\begin{proof}
\begin{enumerate}
 \item First, we check that $R_1\times R_2$ is an object in $\Repr$, i.e. a representation.
  We check both requirements: ${\leqq}$ is a preorder that is sound w.r.t. ${\mmodels}$.
 \begin{itemize}
  \item ${\leqq}$ is a preorder:
  \begin{itemize}
   \item For every preorder $x:A\relation A$ and function $f: B\to A$, $f_\ast\then x\then f^\ast$ is a preorder, therefore both $\pi_{i\ast}\then{\leq}_i\then\pi_i^\ast$ are preorders.
   \item The intersection of two preorders is itself a preorder, so ${\leqq}$ is a preorder.
  \end{itemize}
  \item ${\mmodels}\then{\leqq}\leqrel{\mmodels}$:
  \begin{align*}
   {\mmodels}\then{\leqq}&\eqrel \paren(\bigcup_i\iota_i^\ast\then{\models}_i\then\pi_i^\ast)
   \then{\leqq}\\
   &\eqrel \bigcup_i\iota_i^\ast\then{\models}_i\then\pi_i^\ast\then{\leqq}
   \\
   &\eqrel \bigcup_i\paren(\iota_i^\ast\then{\models}_i\then\pi_i^\ast\then\paren(\bigcap_j\pi_{j\ast}\then{\leq}_j\then\pi_j^\ast))
   \\
   &\leqrel \bigcup_i\iota_i^\ast\then{\models}_i\then\pi_i^\ast\then\pi_{i\ast}\then{\leq}_i\then\pi_i^\ast
   \\
   &\leqrel \bigcup_i\iota_i^\ast\then{\models}_i\then{\leq}_i\then\pi_i^\ast
   \\
   &\leqrel \bigcup_i\iota_i^\ast\then{\models}_i\then\pi_i^\ast
    \eqrel{\mmodels} 
  \end{align*}
 \end{itemize}

  Now we need to look at the requirements for being a Cartesian product.
   We only sketch the proof. We first define the projections $\Pi_1$ and $\Pi_2$:
  \begin{align*}
    \Pi^\phi_i:\,&E_1\times E_2\to E_i
    &\Pi^\psi_i:\,&T_i\relation T_1 + T_2\\
    &\Pi^\phi_i\eqdef \pi_i
    &&\Pi^\psi_i\eqdef \iota_{i\ast}
  \end{align*}
  It is easy enough to check that the $\Pi_i$s are indeed morphisms of representations.

  Next, one needs to check that given a representation $R$ and another pair of morphisms $f_i=\tuple<\phi_i,\psi_i>:R\to R_i$, there is a unique morphism $g:R\to R_1\times R_2$ such that for each $i$, $f_i=\Pi_i\circ g$.
  \begin{itemize}
    \item \emph{Existence.} Let $e\in E_R$, we define $g^\phi(e)=\tuple<\phi_1(e),\phi_2(e)>$.
    We also define $g^\psi\eqrel\bigcup_{i\in\set{1,2}}\iota_i^\ast\then\psi_i$.
    $g$ is a homomorphism, and satisfies the required identities $f_i=\Pi_i\circ g$.
    \item \emph{Unicity.} Assume another such morphism. Pose $g'=\tuple<\phi',\psi'>$.
    The equations $\Pi_i\circ g=f_i= \Pi_i\circ g'$ constrain $g'$ enough so we can prove that $g=g'$.
  \end{itemize}
\item
Finally we need to consider the case where both $R_i$ are exact.

  First, observe that since for each $i\in\set{1,2}$, $\iota_i^\ast\then{\models}_i\then\pi_i^\ast\leqrel{\mmodels}$, we have by contravariance of $-\under x$ that:
  \[{\mmodels}\under{\mmodels}\leqrel\paren(\iota_i^\ast\then{\models}_i\then\pi_i^\ast)\under {\mmodels}\]

  Next, we prove that for each $i\in\set{1,2}$, ${\models}_i\then\pi_i^\ast\then{\mmodels}\under{\mmodels}\leqrel{\models}_i\then\pi_i^\ast$.
  \begin{align*}
    {\models}_i\then\pi_i^\ast\then{\mmodels}\under{\mmodels}
    &\leqrel
    \iota_{i\ast}\then\iota_i^\ast\then{\models}_i\then\pi_i^\ast\then(\iota_i^\ast\then{\models}_i\then\pi_i^\ast)\under{\mmodels}\\
    &\leqrel
    \iota_{i\ast}\then{\mmodels}\\
    &\eqrel
    \iota_{i\ast}\then \iota_1^\ast\then{\models}_1\then\pi_1^\ast\cup \iota_{i\ast}\then \iota_2^\ast\then{\models}_2\then \pi_2^\ast\\
    &\eqrel
    1\then{\models}_i\then\pi_i^\ast\cup 0\then{\models}_{\neg i}\then \pi_{\neg i}^\ast\\
    &\eqrel {\models}_i\then\pi_i^\ast
  \end{align*}

  Therefore, we obtain that ${\mmodels}\under{\mmodels}\leqrel({\models}_i\then\pi_i^\ast)\under({\models}_i\then\pi_i^\ast)$ for both values of $i$. Hence we get:
  \[{\mmodels}\under{\mmodels}\leqrel\bigcap_i({\models}_i\then\pi_i^\ast)\under({\models}_i\then\pi_i^\ast).\]
  Using totality and univalence, we know that for any functions $f,g$ and relations $x,y$ of appropriate types, $(x\then f^\ast)\under(y\then g^\ast)\eqrel f_\ast\then(x\under y)\then g^\ast$.
  Hence:
  \[{\mmodels}\under{\mmodels}\leqrel\bigcap_i\pi_{i\ast}\then({\models}_i\under{\models}_i)\then\pi_i^\ast.\]
  Using exactness of the $R_i$s, we conclude:
  \begin{align*}
    &{\mmodels}\under{\mmodels}\leqrel\bigcap_i\pi_{i\ast}\then({\models}_i\under{\models}_i)\then\pi_i^\ast
  \leqrel\bigcap_i\pi_{i\ast}\then{\leq}_i\then\pi_i^\ast\eqrel{\leqq}.\qedhere
  \end{align*}
  \end{enumerate}

\end{proof}

\subsection{Reductions}
Here we consider how one can establish the exactness of a representation.
We give this question special attention because it is usually the difficult part of correctness proofs: while soundness can be established by checking that individual inference rules are safe, completeness requires one to build a proof tree for every true semantic inclusion.
For this later task, it seems no general method is applicable.
We show that in the case of representations, we can provide such a general method, and analyze a few examples from the literature as instances of our construction.
\begin{definition}[reduction]
A reduction from a representation $R_1$ to $R_2$ is a triple $\tuple<\phi,\tau,\psi>:R_1\leadsto R_2$ such that:
 \begin{itemize}
  \item $\phi:E_1\to E_2$ and $\tau:E_2\to E_1$ are functions and $\psi:T_2\relation T_1$ is a relation;
  \item $\tau$ is order preserving: $\tau^\ast\then{\leq}_2\leqrel{\leq}_1\then\tau^\ast$;
  \item the interpretation of $\phi(e)$ can be computed from that of $e\in E_1$ using $\psi$:
  \[{\models}_2\then\phi^\ast\eqrel\psi\then{\models}_1;\]
  \item the composition $\tau\circ\phi$ yields a ${\leq}_1$-equivalent expression:
  \[\tau^\ast\then\phi^\ast\leqrel {\leq}_1\quad\quad\phi_\ast\then\tau_\ast\leqrel {\leq}_1.\]
 \end{itemize}
\end{definition}

Like morphisms, reductions compose, so we may easily check that $\leadsto$ (the binary relation relating $R_1$ and $R_2$ if there is a reduction from from $R_1$ to $R_2$) is a preorder.
Such transformations are very useful to establish exactness, as witnessed by the following fact:
\begin{proposition}
  $R_1$ is exact if and only if there exists an exact representation $R_2$ and a reduction $R_1\leadsto R_2$.
\end{proposition}
\begin{proof}
 Assuming $R_1$ is exact, then $\tuple<\id_{E_1},\id_{E_1},1_{T_1}>$ is a reduction from $R_1$ to $R^{{\models}_1}$, the trivial representation associated with ${\models}_1$.
 Indeed:
 \begin{align*}
  {\id}^\ast\then{\models}_1\under{\models}_1&\eqrel{\models}_1\under{\models}_1\eqrel{\leq}_1\eqrel{\leq}_1\then{\id}^\ast\\
  {\models}_1\then{\id}^\ast&\eqrel{\models}_1\eqrel 1\then{\models}_1\\
  {\id}^\ast\then{\id}^\ast&\eqrel 1\leqrel {\leq}_1\\
  {\id}_\ast\then{\id}_\ast&\eqrel 1\leqrel {\leq}_1
 \end{align*}

 On the other hand, assume $R_2$ is exact and we have $\tuple<\phi,\tau,\psi>:R_1\leadsto R_2$.
 First, we establish that ${\models}_1\under{\models}_1\leqrel\paren({\models}_2\then\phi^\ast)\under\paren({\models}_2\then\phi^\ast)$, using \eqref{eq:galois}:
 \begin{align*}
  {\models}_1\under{\models}_1\leqrel\paren({\models}_2\then\phi^\ast)\under\paren({\models}_2\then\phi^\ast)
\quad&\Leftrightarrow\quad{\models}_2\then\phi^\ast\then{\models}_1\under{\models}_1\leqrel{\models}_2\then\phi_\ast\\
{\models}_2\then\phi^\ast\then{\models}_1\under{\models}_1&\eqrel\psi\then{\models}_1\then{\models}_1\under{\models}_1\\
&\leqrel\psi\then{\models}_1\eqrel{\models}_2\then\phi^\ast.
 \end{align*}
 We then conclude:
\begin{align*}
 {\models}_1\under{\models}_1&\leqrel\paren({\models}_2\then\phi^\ast)\under\paren({\models}_2\then\phi^\ast)\\
 &\eqrel\phi_\ast\then{\models}_2\under{\models}_2\then\phi^\ast
 \tag*{by \eqref{eq:fun-res}}\\
 &\eqrel\phi_\ast\then{\leq}_2\then\phi^\ast
 \tag*{by exactness of $R_2$}\\
 &\leqrel \phi_\ast\then\tau_\ast\then\tau^\ast\then{\leq}_2\then\phi^\ast
 \tag*{by totality of $\tau$}\\
 &\leqrel \phi_\ast\then\tau_\ast\then{\leq}_1\then\tau^\ast\then\phi^\ast\\
 &\leqrel {\leq}_1  \then{\leq}_1\then{\leq}_1\leqrel{\leq}_1.\qedhere
 \end{align*}
\end{proof}

\begin{remark}
 When giving talks about this work, there's been a recurring question: are reductions morphisms of some kind?
 Looking at the types of these objects, there are two morphism candidates to consider :
 $\tuple<\phi,\psi>:R_1\to R_2$ and
 $\tuple<\tau,{\psi\converse}>:R_2\to R_1$.
 The definition of a reduction gives us one of the two conditions for being a morphism in each case, but not the other:
 \begin{align*}
  \phi^\ast\then{\leq}_1&\nleqrel{\leq}_2\then\phi^\ast
  &{\models}_2\then\phi^\ast&\eqrel\psi\then{\models}_1\\
  \tau^\ast\then{\leq}_2&\leqrel{\leq}_1\then\tau^\ast
  &{\models}_1\then\tau^\ast&\neqrel\psi\converse\then{\models}_2
 \end{align*}
 Thus they both fail to be morphisms in general.
 However, in the category $\xRepr$, the pair $\tuple<\phi,\psi>$ is in fact a morphism from $R_1$ to $R_2$.
\end{remark}

Several proof strategies used in the literature may be construed as reductions.
We review some instances from the realm of extensions of Kleene algebra.

\paragraph{Kleene algebra with tests}
We quote from~\cite{kozen_kleene_1997}, in the beginning of section 7, where the authors prove completeness of the axioms of KAT with respect to the guarded strings model:
\begin{quotation}
... every term $p$ can be transformed into a KAT-equivalent term $\hat p$ such that $G(\hat p)$, the set of guarded strings represented by $\hat p$, is the same as $R(\hat p)$, the set of strings represented by $\hat p$ under the ordinary interpretation of regular expressions.
\end{quotation}
This is the crux of their proof, and can easily be cast as a reduction between the representation of KAT and that of KA.
In this case, the function $\hat{\,}$ would play the role of $\phi$, while both $\tau$ and $\psi$ would be the identity.
It is a simple exercise to check that the conditions laid down by Kozen and Smith are equivalent to requiring that $\tuple<\hat{\,},id,1>$ is a reduction.

\paragraph{Synchronous Kleene algebra}
In~\cite{wagemaker_completeness_2019}, again in section 7, we find a completeness proof.
This proof goes into two steps.
First, completeness is established for a subset of expressions.
Second, it is shown that every expression can be translated into this sub-syntax while remaining provably equivalent.
Each of these steps is an instance of reduction.

For the first, we can define a representation whose set of expressions is precisely the sub-syntax in question, and satisfaction and reasoning order are simply the restrictions to this set of the corresponding relations in the general representation SKA.
The syntax of this sub-representation happens to be the set of regular expressions over a particular alphabet.
And so the proof goes on to provide a function $\Pi$ that translates the semantic elements of SKA into words, and to prove that the image by $\Pi$ of the SKA interpretation of an expression coincides with its regular language.
This can be seen as a reduction from the sub-representation to that of KA (over the special alphabet), with both $\phi$ and $\tau$ being the identity, and $\psi$ being the relation $\Pi^\ast$.

For the second step, the translation function, also called $\hat{\,}$, has the property that $\hat e\equiv_{SF_1}e$ and $\hat e$ belongs to the sub-syntax of the previous step.
As in the case of KAT, this would be a reduction $\tuple<\hat{\,},id,1>$ from SKA to its exact sub-representation.

As we described earlier, reductions compose, so the whole proof may be construed as a single reduction.
The result of computing the composition of the two reductions above yields the triple $\tuple<\hat\,,id,\Pi^\ast>$.

\paragraph{Concurrent Kleene algebra}
For CKA, in~\cite{kappe_concurrent_2018} the authors rely on the completeness of weak bi-Kleene algebra (BKA), a formalism with the same syntax.
The proof proceeds by computing a so-called syntactic closure (Definition 4.1):
\begin{quotation}
 Let $e \in T$ . We say that $e\downarrow$ is a closure of $e$ if both $e \equiv_{CKA} e\downarrow$ and $\sem[e\downarrow]_{BKA} = \sem[e ]_{CKA}$ hold.
 \footnote{Some notations have been modified.}
\end{quotation}
The paper goes on to prove that the existence of such a closure function implies completeness (Lemma 4.1), and then builds one (Theorem 4.1).
  The function $\downarrow$ may be understood as playing the role of $\phi$, and in this case we again set both $\tau$ and $\psi$ to the identity, so this is still an instance of reduction between the representation CKA and the representation BKA.

In fact the specificities of this case would allow for weaker proof obligations to obtain completeness.
(Note however that the weaker obligations are still logically equivalent to the above definition, in the presence of the simplifying hypotheses.)
We may summarize what makes this case easier than the general one by the following remarks:
\begin{align*}
 E_{CKA}&=E_{BKA}
 &T_{CKA}&=T_{BKA}\\
 {\leq}_{BKA}&\leqrel{\leq}_{CKA}
 &{\models}_{BKA}&\leqrel{\models}_{CKA}
\end{align*}
In this case, one only needs a function $\downarrow$ such that
\begin{align*}
 {\models}_{CKA}&\leqrel{\models}_{BKA}\then\downarrow^\ast
 &\downarrow^\ast&\leqrel{\leq}_{CKA}
\end{align*}
One can prove that under these hypotheses, these conditions are equivalent to $\tuple<\downarrow,id,1>$ being a reduction.
We state the general result, after the following definition.

\begin{definition}[syntactic closure]\label{def:closure}
 A syntactic closure between two representations sharing the same sets of expressions and traces, written $R_i=\tuple<T,E,{\models}_i,{\leq}_i>$ for $i\in\set{1,2}$, is a function $\downarrow:E\to E$ such that:
\begin{align*}
 {\models}_1&\leqrel{\models}_2\then\downarrow^\ast
 &\downarrow^\ast&\leqrel{\leq}_1
\end{align*}
\end{definition}
\begin{lemma}\label{lem:closure}
 In the situation of Definition~\ref{def:closure}, if additionally we have ${\leq}_2\leqrel{\leq}_1$, ${\models}_2\leqrel{\models}_1$, and $R_2$ exact, then for any function $\downarrow:E\to E$, $\tuple<\downarrow,id,1>$ is a reduction iff
  $\downarrow$ is a syntactic closure.
\end{lemma}
\begin{proof}
  We prove both implications.
 \begin{enumerate}
   \item Let $\downarrow$ be a syntactic closure. We will show that $\tuple<\downarrow,id,1>$ is a reduction.
   \begin{align*}
    &id^\ast\then{\leq}_2\eqrel{\leq}_2\leqrel{\leq}_1.\\
    &{\models}_2\then\downarrow^\ast\leqrel{\models}_1\then{\leq}_1\leqrel 1\then{\models}_1\leqrel{\models}_2\then\downarrow^\ast\\
    &id^\ast\then\downarrow^\ast\eqrel\downarrow^\ast\leqrel{\leq}_1
   \end{align*}
   For the last inclusion, we observe that
   \[{\models}_2\then\downarrow_\ast\leqrel{\models}_1\then\downarrow_\ast\leqrel{\models}_2\then\downarrow^\ast\then\downarrow_\ast\leqrel{\models}_2\]
   Therefore $\downarrow_\ast\leqrel{\models}_2\under{\models}_2$. We conclude by exactness of $R_2$:
   \[\downarrow_\ast\then id_\ast\eqrel\downarrow_\ast\leqrel{\models}_2\under{\models}_2\leqrel{\leq}_2\leqrel{\leq}_1.\]
   \item Let $\downarrow$ be such that $\tuple<\downarrow,id,1>$ is a reduction. Then:
   \begin{align*}
   &{\models}_1\eqrel 1\then{\models}_1\eqrel{\models}_2\then\downarrow^\ast\\
   &\downarrow^\ast\eqrel id^\ast\then\downarrow^\ast\leqrel{\leq}_1.\qedhere 
   \end{align*}
  \end{enumerate}

\end{proof}

In particular, if $R_2$ is exact, the existence of a syntactic closure entails the exactness of $R_1$.

\newpage
\section{Higher order representations}
\label{sec:hor}
We turn our sights on set-indexed representations, structures intended to model e.g. the free monoid over some set of generators.
These may be seen as functors from \Set to \Repr.
To better understand them, we must investigate competing notions of \emph{naturality} (as in natural transformations) for binary relations.

\subsection{Natural relations}
\newcommand\F{\mathbf F}
\newcommand\G{\mathbf G}
Let $\F,\G$ be a pair of \Set endofunctors, and $\rho$ be a set-indexed family of relations such that $\rho_A:\F A\relation \G A$.
A first notion of naturality can be formulated as follows:
\begin{definition}[natural relation]
 $\rho$ is called \emph{natural} iff for every function $f:A\to B$ and every pair $\tuple<a,a'>\in\rho_A$, we have $\tuple<\F f(a),\G f(a')>\in\rho_B$. Relationally:
 \[\F f^\ast\then\rho_A\leqrel\rho_B\then\G f^\ast.\]
\end{definition}
This states that the relation $\rho$ is preserved by functor application.
This rather ad-hoc definition will be the weakest (most permissive) notion of naturality for relations. We denote a natural relation by writing $\rho:\F\natrel\G$.

Examples of such relations include for instance preorders ${\leq}_{\E}$ over a term algebra $T_\Sigma(A)$ generated from a set of axioms $\E:T_\Sigma(\V)\relation T_\Sigma(\V)$ (where $\V$ is a fixed set of variables).

To define our second notion of naturality, we will use the canonical functor from \Set to \Rel, mapping objects to themselves and arrows $f$ to $f_\ast$. We denote the composition of a \Set-endofunctor $\F$ with it by $\F_\ast$, meaning $\F_\ast A=\F A$ and $\F_\ast f=(\F f)_\ast$.
Now we define:
\begin{definition}[left-linear relation]
 A \emph{left linear} relation from $\F$ to $\G$ is a natural transformation $\F_\ast\Rightarrow\G_\ast$. Relationally, for every function $f:A\to B$:
 \[\F f_\ast\then\rho_B\eqrel\rho_A\then\G f_\ast.\]
\end{definition}

What does this have to do with linearity (or leftism for that matter)? To understand we need relation liftings~\cite{kurz_relation_2016}.
We don't have the space to properly define these (and refer the interested reader to \textit{op. cit.}), so we shall only mention that whenever a \Set endofunctor $\F$ preserves weak pullbacks, it extends to a monotone \Rel endofunctor $\bar\F$ such that $\bar \F(f_\ast)=\F f_\ast$.
Examples of such functors abound.
Most functors corresponding to datatypes fall into that category, like lists, trees, finite graphs...
We shall be especially interested in the following functors, all of which share this property:
\begin{itemize}
 \item $T_\Sigma$, the terms over the signature $\Sigma$;
 \item $\pset$, the powerset functor;
 \item $\List$, the list functor;
 \item $\Pom$, the pomset functor.
\end{itemize}
In the following, we assume both $\F$ and $\G$ preserve weak pullbacks.

We now make the following observation:
\begin{lemma}
 $\rho$ is left-linear iff for every relation $x:A\relation B$ we have:
 \[\bar \F x\then \rho_B\leqrel\rho_A\then\bar\G x\]
\end{lemma}
\begin{proof}
 Assume for every $x:A\relation B$ we have $\bar \F x\then \rho_B\leqrel\rho_A\then\bar\G x$, and let $f:A\to B$.
 Then in particular we do have
 \[\F f_\ast\then\rho_B\eqrel\bar \F(f_\ast)\then\rho_B\leqrel\rho_A\then\bar\G(f_\ast)\eqrel\rho_A\then\G f_\ast.\]
 For the converse inclusion, we use totality and univalence:
 \[\rho_A\then\G f_\ast
   \leqrel\F f_\ast\then\F f^\ast\then\rho_A\then\G f_\ast
   \leqrel\F f_\ast\then\rho_B\then\G f^\ast\then\G f_\ast
   \leqrel\F f_\ast\then\rho_B.
   \]

   Now, assume $\rho$ is left-linear, and let $x:A\relation B$. Seeing $x$ as a set of pairs, we write $dx:x\to A$ for the first projection restricted to $x$ and $cx:x\to B$ for the second. This allows one to see the relation $x$ as the span $A\xleftarrow{dx}x\xrightarrow{cx} B$, and we get $x\eqrel dx^\ast\then cx_\ast$. Therefore, we obtain:
   \begin{align*}
\bar \F x\then \rho_B
   \eqrel \bar \F dx^\ast\then\bar\F cx_\ast\then\rho_B
   &\eqrel \bar \F dx^\ast\then\rho_x\then\bar\G cx_\ast\\
   &\leqrel \bar \F dx^\ast\then\rho_x\then\bar \G dx_\ast\then\bar \G dx^\ast\then\bar\G cx_\ast\\
   &\eqrel \bar \F dx^\ast\then\bar \F dx_\ast\then\rho_A\then\bar \G dx^\ast\then\bar\G cx_\ast\\
   &\leqrel \rho_A\then\bar \G dx^\ast\then\bar\G cx_\ast
   \eqrel\rho_A\then\bar\G x\qedhere
   \end{align*}
\end{proof}
This invites the notion of right-linearity: $\rho$ is \emph{right-linear} iff for every $x:A\relation B$ we have $\rho_A\then\bar \G x\leqrel\bar \F x\then\rho_B$.
Equivalently, $\rho$ is right-linear iff for every function $f:A\to B$ we have $\F f^\ast\then\rho_A\eqrel\rho_B\then\G f^\ast$.
Note that $\rho$ is right-linear iff $\rho\converse$ is left-linear.

An important example of right-linear relation is the membership relation $\in:Id\natrel \pset$ (where $Id$ is the identity functor).
Indeed, for a given $f:A\to B$, the definition of the powerset functor tells us that $\pset f(X)=\setcompr{f(x)}{x\in X}$. This can be reformulated relationally as:
\[{\in}\then \pset f^\ast\eqrel f^\ast\then{\in}\]
which means $\in$ is right-linear.

Finally, we define full linearity:
\begin{definition}[linear relation]\label{def:linear}
A \emph{linear} relation is a natural transformation from $\bar \F$ to $\bar \G$. Relationally, for every relation $x:A\relation B$ we have:
 \[\bar \F x\then \rho_B\eqrel\rho_A\then\bar\G x.\]
\end{definition}
Observe that a relation is linear iff it is both left- and right-linear.

To get an example of linear relation, consider terms over a signature $\Sigma$. Let $\ell:T_\Sigma\Rightarrow List$ be the natural transformation that extracts from a term its list of variables. For instance:
\begin{align*}
 &\ell(f(g(a,b),h(c)))=[a;b;c]
 &&\ell(g(a,f(h(b),c)))=[a;b;c]
\end{align*}
We can define an equivalence relation $\approx$ over terms as $\ell_\ast\then\ell^\ast$, i.e. we relate two terms if they share the same list of variables.
It is a simple exercise to check that $\approx:T_\Sigma\natrel T_\Sigma$ is a linear relation.

To complete this subsection, we turn our attention to natural transformations between $\F$ and $\G$:
\begin{lemma}
 Let $\phi$ be a set-indexed family of functions such that for every set $A$ $\phi_A:\F A\to\G A$. $\phi$ is a natural transformation iff $\phi^\ast$ is right-linear iff $\phi_\ast$ is left-linear.
\end{lemma}
This observation in mind, we define:
\begin{definition}[linear transformation]
 A natural transformation $\phi$ is called \emph{linear} iff $\phi_\ast$ (or equivalently $\phi^\ast$) is a linear relation.
\end{definition}

In the previous example, the transformation $\ell$ is linear. This fact entails the linearity of $\approx\eqrel\ell_\ast\then\ell^\ast$.

Important examples of linear transformations arise from syntax monads.
Indeed, for any functional signature $\Sigma$, it is well known that we can endow the functor $T_\Sigma$ with a monad structure $\tuple<T_\Sigma,\mu_T,\eta_T>$.
It so happens that whatever the signature, the natural transformations $\eta_T$ and $\mu_T$ are linear.
This is not the case e.g. with the powerset monad, where $\mu_P$ fails to be linear.

\subsection{HORs}
We are now equipped to define a structure that tries to capture set-indexed collections of representations, that enjoy some regularity properties.
\begin{definition}[HOR]
 A \emph{higher-order representation} (HOR for short) is a structure $\R=\tuple<\T,\E,{\mmodels},{\leqq}>$ where $\E$ and $\T$ are a \Set endofunctors, ${\mmodels}:\T\natrel\E$ is a \emph{right-linear} relation, and ${\leqq}:\E\natrel\E$ is a \emph{natural} relation such that for every set $A$, $\R_A=\tuple<\T A,\E A,{\mmodels}_A,{\leqq}_A>$ is a representation.
 It is called \emph{relational} when both $\E$ and $\T$ preserve weak pullbacks.
\end{definition}
Note that requiring ${\mmodels}$ to be right-linear is logically equivalent to the interpretation function $I$ being a natural transformation $\E\Rightarrow \pset{\T}$.

These can be turned into functors $\Set\to\Repr$. Indeed, we extend $\R$ to morphisms by setting for a function $f:A\to B$:
\[\R f=\tuple<\E f,\T f^\ast>\]
\begin{lemma}
  The transformation defined on a set $A$ as $\R_A$ and on a function $f$ as $\R f$ is a functor $\Set\to\Repr$.
\end{lemma}
\begin{proof}
  Most properties follow directly from functoriality of $\E$ and $\T$.
 The only non-trivial thing to check that $\R$ is a functor is that $\R f$ is a morphism from $\R A$ to $\R B$:
 \begin{align*}
  \E f^\ast\then{\leqq}_A&\leqrel{\leqq}_B\then \E f^\ast\tag*{naturality of ${\leqq}$}\\
  {\mmodels}_B\then\E f^\ast&\eqrel\T f^\ast\then{\mmodels}_A\tag*{right-linearity of ${\mmodels}$}
 \end{align*}
\end{proof}

The question remains, are there other functors of the same type, that fail to be HORs?
We may answer in the affirmative, as we offer the following characterization:
\begin{proposition}\label{prop:functors}
 The functors $\Set\to\Repr$ are in one-to-one correspondence with structures $\tuple<\T,\E,{\mmodels},{\leqq}>$ where $\T:\Set\to\Rel$ and $\E:\Set\to\Set$ are functors, ${\mmodels}:\T\natrel\E$ and ${\leqq}:\E\natrel\E$ are natural relations, such that:
 \begin{mathpar}
  {\leqq}\eqrel{\leqq}\under{\leqq}\and
  {\mmodels}\then{\leqq}\leqrel{\mmodels}\and
  \forall f:A\to B,\,{\T f}\then {{\mmodels}}\eqrel{\mmodels}\then\E f^\ast.
 \end{mathpar}
\end{proposition}
Observe that the difference with a HOR is that $\T$ is a functor from \Set to \Rel instead of one from \Set to itself.
As a result, the condition on ${\mmodels}$ differs slightly from right-linearity, and the characterization via interpretations is less clear as the composite functor $\pset \T$ is ill-defined.

\subsection{Lifting HORs}

Relational HORs can be first extended as functors from the category \PreO of preordered sets to \Repr:
\begin{definition}[$\tilde\R$]
  Given a relational HOR $\R$ and a preorder $\tuple<A,{\leq}_A>$ we define:
  \[\tilde\R{\tuple<A,{\leq}_A>}\eqdef\tuple<\T A,\E A,\paren(\bar\T{{\leq}_A}\then{\mmodels}),\paren(\bar\E{{\leq}_A}\cup{\leqq})^\star>.\]
\end{definition}
\begin{lemma}
 $\tilde\R$ is a functor $\PreO\to\Repr$.
\end{lemma}
\begin{proof}
 $\paren(\bar\E{{\leq}_A}\cup{\leqq})^\star$ is quite obviously a preorder, so we check soundness. Thanks to the laws of relation algebra (specifically the KA fragment) this amounts to checking the following pair of inclusions:
 \begin{align*}
  \bar\T{{\leq}_A}\then{\mmodels}\then\bar\E{{\leq}_A}&\leqrel\bar\T{{\leq}_A}\then\bar\T{{\leq}_A}\then{\mmodels}\eqrel\bar\T{\paren({\leq}_A\then{\leq}_A)}\then{\mmodels}\leqrel\bar\T{{\leq}_A}\then{\mmodels}\\
  \bar\T{{\leq}_A}\then{\mmodels}\then{\leqq}&\leqrel\bar\T{{\leq}_A}\then{\mmodels}.\qedhere
 \end{align*}
 \end{proof}
 \paragraph{Monoids over an ordered set}
 This lifting works very well in the case of the free monoid over an ordered set of generators.
 Indeed, let us define the following HOR: $\Mon=\tuple<\List,M,I^\ast,=_{\Mon}>$, where $\List$ is the list functor, $M$ is the term functor $T_{\set{\otimes^{(2)},\unit^{(0)}}}$, $I$ is the natural transformation mapping $\unit$ to the empty list, and products $u\otimes v$ to concatenations $I(u)\cdot I(v)$, and $=_{\Mon}$ is the least congruence on $M(A)$ that contains the following axioms, for every term $u,v,w\in M(A)$:
 \begin{mathpar}
  u\otimes(v\otimes w)=_{\Mon}(u\otimes v)\otimes w\and
  u\otimes\unit=_{\Mon}u=_{\Mon}\unit\otimes u
 \end{mathpar}
 Applying the above construction to $\Mon$ over some preorder $\tuple<A,\trleq>$ yields the structure
\[\tilde \Mon{\tuple<A,\trleq>}\eqdef\tuple<\List\,A,MA,\paren(\bar {\List}{\trleq}\then I^\ast),\paren(\bar M{\trleq}\cup =_{\Mon})^\star>.\]
We may further understand the satisfaction by realizing that $\bar{\List}{\trleq}$ is simply the pointwise application of the ordering between lists of the same length. In other words:
\begin{mathpar}

 [a_1,\dots,a_n]\mathrel{\bar{\List}{\trleq}}[b_1,\dots,b_m]\and
 \Leftrightarrow\and
 n=m\wedge \forall i,a_i\trleq b_i.
\end{mathpar}
  The preorder on expressions may also be reformulated as the least preorder $\tilde{\leq}$ such that:
  \begin{mathpar}
   \inferrule{u\mathrel{\tilde{\leq}} u'\and v\mathrel{\tilde{\leq}} v'}{u\otimes v\mathrel{\tilde{\leq}} u'\otimes v'}\and
   \inferrule{u =_{\Mon} v}{u\mathrel{\tilde{\leq}} v}\and
   \inferrule{a \trleq  b}{a\mathrel{\tilde{\leq}} b}
  \end{mathpar}

 Further, we may extend relational HORs as functors $\Repr\to\Repr$, thus yielding representations parametrized by other representations.
 The result is a two-tiered system, where the HOR specifies the outer (dynamic) part of the model (e.g. trace semantics of regular expressions), and the parameter handles the inner (static) part of the model (e.g. boolean algebra to denote memory states as bit-vectors).
 \begin{definition}[$\hat\R$]
  Given a relational HOR $\R$ and representation $R$, we define
  \[\hat\R(R)\eqdef\tuple<\T T,\E E,\paren(\T{{\models}}\then{\mmodels}),\paren(\E{{\leq}}\cup{\leqq})^\star>.\]
 \end{definition}
 \begin{lemma}
  $\hat\R$ is a functor $\Repr\to\Repr$.
 \end{lemma}

The caveat is that these construction do not automatically preserve exactness.
Indeed in many cases, even if $\R$ was exact (as in, every $\R A$ is exact), $\hat\R$ may need more axioms to be exact.

\paragraph{Application to Kleene algebra}
A number of extensions of Kleene algebra can be construed as HORs.
In a nutshell, those extensions that do not depend on their alphabet in an essential way give rise to well-behaved HORs.
For instance, KA itself is a very neat HOR, which can be defined like so:
\[\mathrm{KA}(A)=\tuple<\List\,A,\Reg(A),\paren({\in}\then L^*),{\leq}_{\mathrm{KA}}>\]
(where $L:\Reg(A)\to \pset{\List\,A}$ is the function mapping an expression to the regular language it denotes.)
$L$ is easily shown to be a natural transformation, and ${\leq}_{\mathrm{KA}}$ is a natural relation if we take Kozen's axiomatization (which is stable under substitutions).

We may frame biKA~\cite{laurence_completeness_2014} in the same way, and for the same reasons.
Perhaps more surprising, CKA~\cite{kappe_concurrent_2018} is amenable to the same treatment, by setting:
\[\mathrm{CKA}(A)=\tuple<\Pom\,A,\mathrm{bi}\Reg(A),\paren({\sqsubseteq}\then{\in}\then L^*),{\leq}_{\mathrm{CKA}}>\]
(where $L$ is the biKA interpretation function, mapping regular expressions with~$\parallel$ to sets of series-parallel pomsets, and $\sqsubseteq$ is the subsumption order.)
The fact that subsumption happens to be a linear relation (as in Definition~\ref{def:linear}) makes the CKA semantic relation right-linear, and thus gives us a well-behaved HOR.

\newpage
\section{Conclusions}
\label{sec:concl}
\subsection{Related work}
An important related theory is that of institutions~\cite{goguen_institutions_1992}.
In institutions one also encapsulates in the same structure some syntactic and some semantic information.
But when a representation tries to capture one system of reasoning, an institution contains many systems.
For that reason the definition of an institution is more demanding than ours: it requires syntax and semantics to be equipped with category structures, as well as compatibility conditions between them.
A legitimate question is then what are the classes of representations that can be turned into institutions?
Or, conversely, can representations be extracted from an institution?
We leave these questions open.

Specification  theories~\cite{fahrenberg_linear-timebranching-time_2020}, introduced by Fahrenberg and Legay, is another related formalism. (A similar formalism was introduced (but not named) in~\cite{aceto_when_2019}.)
A specification theory may be understood as a structure $\tuple<T,E,\chi,{\leq}>$, where $\chi:T\to E$ mapping a trace to its \emph{characteristic expression}.
It can be turned into a representation by setting ${\models}\eqrel\chi_\ast\then{\leq}$, i.e. by defining $I(e)=\setcompr{t\in T}{\chi(t)\leq e}$.
Specification theories are thus special instances of representations, where traces may be represented as expressions, and where satisfaction is derived from the preorder on expressions.
In~\cite{fahrenberg_linear-timebranching-time_2020}, it was shown that specification theories capture a number of interesting formalisms.
We plan on exploring how the representation point of view can help in this context in a future paper.

\subsection{On the meta-meta-theory}
In this paper, our meta-theory of SAS is built inside the categories \Set and \Rel.
However, the keen reader will recognize that we use an axiomatic presentation of \Rel in our proofs, rather than one based on the definition of operations on sets of pairs.
This suggests generalizations of this work may be within reach, by abstracting away \Rel for a more general category equipped with the appropriate operators.
Unfortunately, it seems that all existing formalisms (that the author is aware of) that are expressive enough have extraneous assumptions.
We review some below.

Most participants of RAMiCS would for instance suggest using Relation Algebra, as axiomatised by Tarski~\cite{tarski_calculus_1941}, as a meta theory.
This works nicely enough, but as we do not use negations, and since these are well-recognized sources of complexity in proof systems, we would rather consider a more parsimonious model.
The question then boils down to which is the ``correct'' fragment of Relation Algebra that is necessary here.

In that direction, a strong contender is the notion of Allegory~\cite{bird_algebra_1997,freyd_categories_1990}.
This seems better aligned with our needs, and indeed our theory could be largely formulated in the language of e.g. locally complete allegories or division allegories.
However, these require assumptions that we do not use, and that seem quite demanding.
Mainly, any allegory is required to satisfy the \emph{modular law} (a.k.a. Dedekind's rule).
This requires a particular relationship between composition, intersection and converse that plays no role in our development.

We plan to investigate these questions more thoroughly in a future paper, by defining a new kind of abstract structure that captures exactly our requirements without extra hypotheses.
With any luck, interesting instances of these structures will yield novel systems of analysis, or at least novel ways to study them.

\subsection{Directions for future research}

This paper is intended to be the first step of a broader research program.
Beyond enhancing and consolidating the results in this paper, we aim to apply this theory to concrete verification tasks.
We hope such an endeavour will help build useful representations, and guide us towards more relevant constructions on representations.
This activity will be accompanied by a classification effort, to identify useful classes of representations.
We will also strive to formalize our results in the Rocq proof system.

\pagebreak
\bibliographystyle{plain}
\bibliography{biblio_ramics2026}

\appendix



\end{document}